\documentclass[11pt]{article}

\usepackage[margin=1in]{geometry}
\usepackage{amsmath,amssymb,amsthm}
\usepackage{parskip}
\usepackage[colorlinks=true,linkcolor=blue,citecolor=blue,urlcolor=blue]{hyperref}

\newtheorem{theorem}{Theorem}[section]
\newtheorem{lemma}[theorem]{Lemma}

\theoremstyle{remark}
\newtheorem{remark}[theorem]{Remark}

\title{Minimum Edge-Outerplanar Embeddings are Polynomial-Time Computable}

\author{
Hantao Yu\\
Columbia University\\
\texttt{hantao.yu@columbia.edu}
}
\date{}

\begin{document}

\maketitle

\begin{abstract}
We prove that the minimum edge-outerplanarity of a planar graph can be computed in polynomial time, resolving an open problem of Bentz (2009). The proof was initially produced by GPT~5.5~Pro and then verified and polished manually.
\end{abstract}

\section{Introduction}

A plane embedding of a graph is called \emph{$k$-edge-outerplanar} if the following \emph{edge-peeling process} deletes all edges in at most $k$ rounds: in each round, delete every edge lying on the outer face of the current embedding. The associated optimization problem asks for a planar embedding minimizing the number of edge-peeling rounds.

For the vertex analogue, namely ordinary $k$-outerplanarity, minimum-$k$ embeddings are known to be computable in polynomial time~\cite{BM90}. Bentz introduced $k$-edge-outerplanar graphs in the study of edge-disjoint paths and multicut problems, and asked whether the analogous statement holds for $k$-edge-outerplanarity~\cite{Bentz09}. We answer this question affirmatively for loopless planar graphs.

\begin{theorem}[Main theorem]\label{thm:main}
For every finite loopless planar graph $G$, a planar embedding of $G$ with minimum possible edge-outerplanarity can be found in polynomial time.
\end{theorem}

\begin{remark}\label{rem:multigraphs}
The proof below allows parallel edges. Thus it applies to all finite loopless planar multigraphs, and in particular to all finite simple planar graphs. We do not discuss loops.
\end{remark}

We reduce the problem to the known polynomial-time problem of computing an embedding of minimum face-depth. The \emph{depth} of a plane embedding is the maximum distance, in the face-adjacency graph, from the outer face to any other face. Bienstock and Monma studied polynomial-time algorithms for minimizing several such embedding distance measures~\cite{BM90}. Angelini, Di Battista, and Patrignani later gave an $O(n^4)$-time algorithm for computing a minimum-depth embedding of an $n$-vertex planar graph~\cite{ADP11}.

\subsection{Statement on AI use}

The proof is generated using the agentic pipeline from the pipeline-math project \footnote{https://github.com/Pengbinghui/pipeline-math.git}. The pipeline involves the use of GPT 5.5 Pro and Claude Opus 4.8, where GPT 5.5 Pro acts as a solver and Claude Opus 4.8 acts as a verifier. The proof is then verified by the authors, and thus the authors are solely responsible for the correctness of the proof.

\section{Preliminaries}

Throughout this paper, we assume that all graphs are finite. Unless stated otherwise, graphs are allowed to have parallel edges but have no loops.

Let $\Gamma$ be a plane embedding of a graph $G$. Let $F(\Gamma)$ be the set of faces of $\Gamma$, and let $f_\infty$ be the outer face. The \emph{face-adjacency graph} of $\Gamma$, denoted $D_\Gamma$, is the graph whose vertices are the faces of $\Gamma$, with an edge between two faces whenever they share a primal edge. We also define $F_\Gamma(e)$ to be the set of faces in $F(\Gamma)$ that are incident to $e$. If a primal edge $e$ is a bridge, then it is incident with the same face on both sides; in this case $F_\Gamma(e)$ has one element, and the bridge contributes only a loop in the face-adjacency graph, and such loops are ignored for distances.

For a face $f \in F(\Gamma)$, define
\[
d_\Gamma(f) := \operatorname{dist}_{D_\Gamma}(f_\infty, f).
\]
The \emph{depth} of the embedding $\Gamma$ is
\[
\operatorname{depth}(\Gamma) := \max_{f \in F(\Gamma)} d_\Gamma(f).
\]
Including the outer face in this maximum is harmless, since $d_\Gamma(f_\infty) = 0$.

We note that the face-adjacency graph of any plane graph, connected or not, is connected: a generic arc from an interior point of a face $f$ to a point of $f_\infty$, chosen transversal to the edges and avoiding the vertices, crosses only edges and thus traces a walk from $f$ to $f_\infty$ in $D_\Gamma$. Hence $d_\Gamma(f)$ is finite for every face $f$, and the edge-peeling process deletes every edge after finitely many rounds.

For an edge $e \in E(G)$, let $\lambda_\Gamma(e)$ be the round in which $e$ is deleted by the edge-peeling process. Thus an edge lying on the original outer face has $\lambda_\Gamma(e) = 1$. Define
\[
\operatorname{eop}(\Gamma) := \max_{e \in E(G)} \lambda_\Gamma(e),
\]
with the convention that $\operatorname{eop}(\Gamma) = 0$ if $E(G) = \emptyset$.

Define
\[
\mathrm{OPT}_{\mathrm{edge}}(G) := \min_\Gamma \operatorname{eop}(\Gamma),
\]
where the minimum is over all plane embeddings of $G$, including the choice of the outer face. Similarly, for a planar graph $K$, define
\[
\mathrm{OPT}_{\mathrm{depth}}(K) := \min_\Delta \operatorname{depth}(\Delta),
\]
where the minimum is over all plane embeddings of $K$, including the choice of the outer face.

In the proof, we will use the following known theorem, which shows that computing a minimum-depth embedding of an $n$-vertex planar graph takes at most $O(n^4)$ time.

\begin{theorem}[Angelini, Di Battista, Patrignani~\cite{ADP11}]\label{thm:adp}
Given an $n$-vertex planar graph $K$, one can compute in $O(n^4)$ time a plane embedding of $K$ with minimum possible depth.
\end{theorem}

\section{Edge peeling in a fixed embedding}

We first prove the exact relationship between edge-peeling rounds and face-depths in a fixed embedding.

\begin{lemma}[Face reachability after edge deletions]\label{lem:reachability}
Let $\Gamma$ be a plane embedding of a graph $G$, and let $S \subseteq E(G)$. Delete the edges of $S$, while retaining all vertices. Then the current face containing the original outer face is obtained by merging exactly those original faces that are reachable from $f_\infty$ in the subgraph of $D_\Gamma$ whose edges are dual to the primal edges in $S$.
\end{lemma}

\begin{proof}
Deleting a non-bridge edge glues the two face-regions incident with that edge along the open arc formerly occupied by the edge. Deleting a bridge removes a slit lying in a single face-region and does not merge two distinct original faces. Therefore, after deleting all edges of $S$, the resulting face-regions are precisely the equivalence classes generated by gluing original faces across deleted non-bridge edges.

This is exactly connectivity in the subgraph of $D_\Gamma$ whose edges correspond to the deleted primal edges. Possible identifications through a cutvertex do not create additional equivalences: locally around a vertex, moving from one original face sector to another requires passing across an incident edge sector, and such a passage is available precisely when the corresponding edge has been deleted. Hence the original faces merged with $f_\infty$ are exactly the faces reachable from $f_\infty$ through dual edges corresponding to $S$.
\end{proof}

\begin{lemma}[Fixed-embedding edge layer formula]\label{lem:layer}
Let $\Gamma$ be a plane embedding of a loopless graph $G$. Then, for every edge $e \in E(G)$,
\[
\lambda_\Gamma(e) = 1 + \min_{f \in F_\Gamma(e)} d_\Gamma(f).
\]
Consequently,
\[
\operatorname{eop}(\Gamma) = \max_{e \in E(G)} \Bigl( 1 + \min_{f \in F_\Gamma(e)} d_\Gamma(f) \Bigr).
\]
\end{lemma}

\begin{proof}
Let $S_r$ be the set of edges deleted after the first $r$ peeling rounds, with $S_0 = \emptyset$. Let
\[
B_r := \{ f \in F(\Gamma) : d_\Gamma(f) \le r \}.
\]
We prove by induction on $r$ that, after $r$ rounds, the current outer face is obtained by merging exactly the original faces in $B_r$.

For $r = 0$, this is immediate, since before any edge is deleted the outer region is exactly the original outer face $f_\infty$, and $B_0 = \{f_\infty\}$.

Assume the statement holds after $r$ rounds. At the beginning of round $r+1$, an undeleted edge $e$ lies on the current outer face if and only if at least one original face incident with $e$ lies in the current outer region. By the induction hypothesis, this is equivalent to
\[
\min_{f \in F_\Gamma(e)} d_\Gamma(f) \le r.
\]
Thus the edges deleted in round $r+1$ are exactly the previously undeleted edges satisfying this inequality.

After deleting these edges, Lemma~\ref{lem:reachability} shows that the new outer region consists of the original faces reachable from $B_r$ by crossing one of the newly deleted edges. Hence it contains every face at distance at most $r+1$ from $f_\infty$.

Conversely, if an edge is deleted in round $r+1$, then it has an incident face of depth at most $r$. Crossing such an edge can only reach a face of depth at most $r+1$, because adjacent faces in $D_\Gamma$ have depths differing by at most one. Therefore no original face of depth greater than $r+1$ enters the outer region after round $r+1$.

Thus, after $r+1$ rounds, the current outer face is obtained by merging exactly the original faces in $B_{r+1}$. The induction is complete.

It follows that an edge $e$ is first deleted in the unique round $r+1$ for which at least one incident face has depth at most $r$. Equivalently,
\[
\lambda_\Gamma(e) = 1 + \min_{f \in F_\Gamma(e)} d_\Gamma(f).
\]
Taking the maximum over all edges gives the formula for $\operatorname{eop}(\Gamma)$.
\end{proof}

\section{The auxiliary graph}

Let $G$ be a connected loopless planar graph with at least one edge. We construct an auxiliary graph $H = H(G)$ as follows.

For every edge $e = uv \in E(G)$, first subdivide $e$, replacing it by the path
\[
u - s_e - v.
\]
Then attach a pendant triangle to $s_e$. That is, add three new vertices
\[
t_e, \quad a_e, \quad b_e
\]
and four new edges
\[
s_e t_e, \quad t_e a_e, \quad a_e b_e, \quad b_e t_e.
\]
The cycle
\[
C_e = t_e a_e b_e t_e
\]
is called the \emph{marker triangle} for $e$. The edge $s_e t_e$ is a bridge attaching this marker triangle to the subdivision vertex $s_e$.

The construction is linear:
\[
|V(H)| = |V(G)| + 4|E(G)|, \qquad |E(H)| = 6|E(G)|.
\]
Moreover, if $G$ is planar, then $H$ is planar. If $G$ is loopless, then $H$ is simple even when $G$ has parallel edges.

\section{From edge layers to marker depths}

We first prove that every embedding of $G$ gives an embedding of $H$ of the same depth as the edge-outerplanarity of the embedding of $G$.

\begin{lemma}\label{lem:forward}
Let $G$ be connected and loopless, with at least one edge, and let $H = H(G)$. For every plane embedding $\Gamma$ of $G$, there exists a plane embedding $\Delta$ of $H$ such that
\[
\operatorname{depth}(\Delta) = \operatorname{eop}(\Gamma).
\]
Consequently,
\[
\mathrm{OPT}_{\mathrm{depth}}(H) \le \mathrm{OPT}_{\mathrm{edge}}(G).
\]
\end{lemma}

\begin{proof}
Fix a plane embedding $\Gamma$ of $G$. For every edge $e \in E(G)$, choose an incident face $f_e \in F_\Gamma(e)$ of minimum depth:
\[
d_\Gamma(f_e) = \min_{f \in F_\Gamma(e)} d_\Gamma(f).
\]
Starting from $\Gamma$, subdivide each edge $e$ by inserting $s_e$. This does not change the set of old faces or their distances from the outer face. Now draw the marker triangle $C_e$ inside the chosen face $f_e$, attached to $s_e$ by the bridge $s_e t_e$. Since there are only finitely many such attachments, they can be drawn pairwise disjointly inside their chosen faces. Call the resulting embedding of $H$ $\Delta$.

Let $\tau_e$ be the empty triangular face bounded by $C_e$. The other side of $C_e$, namely the side containing the bridge $s_e t_e$, is the old face $f_e$ modified by the attachment. We first observe that
\[
d_\Delta(f) = d_\Gamma(f) \qquad \text{for every old face } f.
\]
Indeed, the only edges of $D_\Delta$ that are not edges of $D_\Gamma$ are the edges between a face $f_e$ and its marker face $\tau_e$; since each $\tau_e$ is a leaf of $D_\Delta$, whose unique neighbor is $f_e$, these new edges create no shortcuts between old faces, and the outer face is unchanged.

Hence $\tau_e$ is a leaf face one dual step farther from $f_\infty$ than $f_e$. Therefore
\[
d_\Delta(\tau_e) = d_\Gamma(f_e) + 1 = 1 + \min_{f \in F_\Gamma(e)} d_\Gamma(f) = \lambda_\Gamma(e),
\]
where the last equality follows from Lemma~\ref{lem:layer}. Note that this uses the fact that we chose to put the triangle gadget on the side of $e$ with minimum face-depth in $\Gamma$.

It remains to check that no old face has depth larger than the maximum marker depth. Let $f$ be an old face of depth $q > 0$ in $\Gamma$. Let
\[
f_\infty = f_0, f_1, \ldots, f_q = f
\]
be a shortest path in $D_\Gamma$. The final step of this path crosses some original edge $e$ incident with $f_{q-1}$ and $f_q = f$. Thus
\[
\min_{g \in F_\Gamma(e)} d_\Gamma(g) = q - 1.
\]
By Lemma~\ref{lem:layer},
\[
\lambda_\Gamma(e) = q.
\]
So every old face-depth is witnessed by some edge-deletion layer. The only new faces in $\Delta$ are the marker faces $\tau_e$. Therefore
\[
\operatorname{depth}(\Delta) = \max_{e \in E(G)} d_\Delta(\tau_e) = \max_{e \in E(G)} \lambda_\Gamma(e) = \operatorname{eop}(\Gamma).
\]
Taking the minimum over all embeddings $\Gamma$ of $G$ gives
\[
\mathrm{OPT}_{\mathrm{depth}}(H) \le \mathrm{OPT}_{\mathrm{edge}}(G). \qedhere
\]
\end{proof}

\section{From marker depths back to edge layers}

We now prove the converse direction. The key point is that deleting marker gadgets and suppressing subdivision vertices cannot increase the relevant face distances.

\begin{lemma}[Projection of face distances]\label{lem:projection}
Let $G$ be connected and loopless, with at least one edge, and let $H = H(G)$. Let $\Delta$ be a plane embedding of $H$. Let $\Gamma$ be the plane embedding of $G$ obtained from $\Delta$ by deleting every marker triangle together with its attaching bridge, and then suppressing every subdivision vertex $s_e$.

For each face $x$ of $\Delta$, define its image $p(x)$ as follows. If $x$ is not an empty triangular marker face, then $p(x)$ is the face of $\Gamma$ containing the portion of $x$ that remains after the deletions and suppressions. If $x = \tau_e$ is an empty triangular marker face, then set
\[
p(x) = p(\tau_e) = p(\rho_e),
\]
where $\rho_e$ is the face on the bridge side of the marker triangle $C_e$.

Then, for every face $x$ of $\Delta$,
\[
d_\Gamma(p(x)) \le d_\Delta(x).
\]
\end{lemma}

\begin{proof}
Let
\[
x_0, x_1, \ldots, x_\ell = x
\]
be a shortest path in $D_\Delta$ from the outer face $x_0$ of $\Delta$ to $x$. We show that, after deleting consecutive repetitions, the sequence
\[
p(x_0), p(x_1), \ldots, p(x_\ell)
\]
contains a walk in $D_\Gamma$ from the outer face of $\Gamma$ to $p(x)$.

First consider one step $x_i x_{i+1}$ of the dual path. The faces $x_i$ and $x_{i+1}$ share some primal edge $\alpha$ of $H$. Every edge of $H$ is of one of three kinds: an edge of a marker triangle, one of the two halves $u s_e$ or $s_e v$ of a subdivided edge, or an attaching bridge $s_e t_e$. We consider the three kinds in turn.

If $\alpha$ is an edge of a marker triangle, then the two faces incident with $\alpha$ are the empty marker face and the bridge-side face of that marker. These two faces have the same image under $p$, by definition.

If $\alpha$ is one of the two subdivided edges $u s_e$ or $s_e v$, then after suppressing $s_e$ this crossing becomes either a crossing of the restored edge $uv$ in $\Gamma$ or a repetition if the two sides have merged. Thus $p(x_i)$ and $p(x_{i+1})$ are equal or adjacent in $D_\Gamma$.

If $\alpha$ is an attaching bridge $s_e t_e$, then $\alpha$ is a bridge of $H$, so it does not give an edge between two distinct faces in $D_\Delta$. Hence this case cannot occur as a step of the dual path between distinct faces.

Therefore each step of the dual path in $D_\Delta$ maps to either a step or a repetition in $D_\Gamma$.

It remains to identify the starting face. If the outer face $x_0$ of $\Delta$ is not an empty triangular marker face, then $p(x_0)$ is the outer face of $\Gamma$. If $x_0 = \tau_e$ is an empty triangular marker face, then deleting that outer marker triangle opens its bridge-side face $\rho_e$ directly to the exterior, so $p(x_0) = p(\rho_e)$ is again the outer face of $\Gamma$.

Thus the projected sequence gives a walk from the outer face of $\Gamma$ to $p(x)$ of length at most $\ell = d_\Delta(x)$. Hence
\[
d_\Gamma(p(x)) \le d_\Delta(x). \qedhere
\]
\end{proof}

\begin{lemma}\label{lem:backward}
Let $G$ be connected and loopless, with at least one edge, and let $H = H(G)$. For every plane embedding $\Delta$ of $H$, if $\Gamma$ is the embedding of $G$ obtained by deleting all marker gadgets and suppressing all subdivision vertices, then
\[
\operatorname{eop}(\Gamma) \le \operatorname{depth}(\Delta).
\]
Consequently,
\[
\mathrm{OPT}_{\mathrm{edge}}(G) \le \mathrm{OPT}_{\mathrm{depth}}(H).
\]
\end{lemma}

\begin{proof}
Let
\[
h = \operatorname{depth}(\Delta).
\]
Fix an original edge $e = uv \in E(G)$. Its marker triangle is
\[
C_e = t_e a_e b_e t_e.
\]
Since $C_e$ is a simple cycle in the plane, it has two sides. Because the rest of $H$ is attached to $C_e$ only at the vertex $t_e$ through the bridge $s_e t_e$, one side of $C_e$ contains no other part of the graph. This side is an empty triangular face; call it $\tau_e$. Let $\rho_e$ be the face on the other side of $C_e$, the side containing the bridge $s_e t_e$.

There are two cases. First suppose $\tau_e$ is not the outer face of $\Delta$. Then $\tau_e$ is a leaf in the face-adjacency graph $D_\Delta$, and its unique neighbor is $\rho_e$. Hence
\[
d_\Delta(\rho_e) = d_\Delta(\tau_e) - 1 \le h - 1.
\]
By Lemma~\ref{lem:projection},
\[
d_\Gamma(p(\rho_e)) \le d_\Delta(\rho_e) \le h - 1.
\]
After deleting the marker gadget, the face $p(\rho_e)$ is incident with the subdivided path
\[
u - s_e - v.
\]
After suppressing $s_e$, this same face is incident with the restored edge $uv$. Therefore the edge $uv$ has an incident face in $\Gamma$ of depth at most $h - 1$. By Lemma~\ref{lem:layer} we immediately get
\[
\lambda_\Gamma(uv) \le h.
\]

Now suppose $\tau_e$ is the outer face of $\Delta$. Deleting this outer marker triangle opens the bridge-side face $\rho_e$ directly to the exterior. After suppressing $s_e$, the restored edge $uv$ is incident with the outer face of $\Gamma$. Therefore
\[
\lambda_\Gamma(uv) = 1 \le h.
\]
Here $h \ge 1$, because $H$ has at least one marker triangle and hence at least two faces, and the face-adjacency graph of a plane graph is connected, so some face has positive depth.

Thus $\lambda_\Gamma(e) \le h$ for every original edge $e \in E(G)$. Therefore
\[
\operatorname{eop}(\Gamma) \le h = \operatorname{depth}(\Delta).
\]
Taking the minimum over all embeddings $\Delta$ of $H$ gives
\[
\mathrm{OPT}_{\mathrm{edge}}(G) \le \mathrm{OPT}_{\mathrm{depth}}(H). \qedhere
\]
\end{proof}

\begin{theorem}[Equality of optima]\label{thm:equality}
Let $G$ be a connected loopless planar graph with at least one edge, and let $H = H(G)$ be the auxiliary graph constructed above. Then
\[
\mathrm{OPT}_{\mathrm{edge}}(G) = \mathrm{OPT}_{\mathrm{depth}}(H).
\]
\end{theorem}

\begin{proof}
Lemma~\ref{lem:forward} gives
\[
\mathrm{OPT}_{\mathrm{depth}}(H) \le \mathrm{OPT}_{\mathrm{edge}}(G).
\]
Lemma~\ref{lem:backward} gives
\[
\mathrm{OPT}_{\mathrm{edge}}(G) \le \mathrm{OPT}_{\mathrm{depth}}(H).
\]
Therefore the two quantities are equal.
\end{proof}

\section{Disconnected graphs}

We now remove the connectedness assumption.

\begin{lemma}\label{lem:components}
Let $G$ be a loopless planar graph with connected components
\[
G_1, \ldots, G_c.
\]
Then
\[
\mathrm{OPT}_{\mathrm{edge}}(G) = \max_{1 \le i \le c} \mathrm{OPT}_{\mathrm{edge}}(G_i).
\]
Isolated vertices have value $0$.
\end{lemma}

\begin{proof}
For the upper bound, take an optimal embedding $\Gamma_i$ of each component $G_i$, and draw all components side by side, so that their outer faces lie in the common outer face of the resulting embedding $\Gamma$ of $G$. The faces of $\Gamma$ are the internal faces of the embeddings $\Gamma_i$, together with one common outer face $f_\infty$ obtained by merging the outer faces of the $\Gamma_i$. Any path in $D_{\Gamma_i}$ from the outer face of $\Gamma_i$ to an internal face $f$ of $\Gamma_i$ is also a path in $D_\Gamma$ from $f_\infty$ to $f$, so
\[
d_\Gamma(f) \le d_{\Gamma_i}(f) \qquad \text{for every internal face } f \text{ of } \Gamma_i.
\]
By Lemma~\ref{lem:layer}, applied to $\Gamma$ and to each $\Gamma_i$, every edge $e$ of $G_i$ satisfies
\[
\lambda_\Gamma(e) = 1 + \min_{f \in F_\Gamma(e)} d_\Gamma(f) \le 1 + \min_{f \in F_{\Gamma_i}(e)} d_{\Gamma_i}(f) = \lambda_{\Gamma_i}(e).
\]
Hence
\[
\operatorname{eop}(\Gamma) \le \max_i \operatorname{eop}(\Gamma_i) = \max_i \mathrm{OPT}_{\mathrm{edge}}(G_i).
\]

For the lower bound, let $\Gamma$ be any embedding of $G$. Fix a component $G_i$, and delete all other components from the drawing; let $\Gamma_i'$ denote the induced embedding of $G_i$. This operation cannot increase face-distances for the remaining component: any path in $D_\Gamma$ from $f_\infty$ to a face projects to a walk in $D_{\Gamma_i'}$, because a step across an edge of $G_i$ remains a step or becomes a repetition, while a step across an edge of another component becomes a repetition, as the two faces flanking that edge are merged when it is deleted; moreover, the outer face of $\Gamma$ is contained in the outer face of $\Gamma_i'$. Hence, for every edge $e$ of $G_i$, the incident faces of $e$ in $\Gamma_i'$ are the images of its incident faces in $\Gamma$, and their depths in $\Gamma_i'$ are at most the corresponding depths in $\Gamma$. By Lemma~\ref{lem:layer}, applied to $\Gamma_i'$ and to $\Gamma$,
\[
\lambda_{\Gamma_i'}(e) \le \lambda_\Gamma(e) \le \operatorname{eop}(\Gamma) \qquad \text{for every } e \in E(G_i),
\]
so $\operatorname{eop}(\Gamma_i') \le \operatorname{eop}(\Gamma)$, and therefore
\[
\mathrm{OPT}_{\mathrm{edge}}(G_i) \le \operatorname{eop}(\Gamma)
\]
for every $i$. Taking the maximum over $i$, and then minimizing over all embeddings $\Gamma$ of $G$, gives
\[
\max_i \mathrm{OPT}_{\mathrm{edge}}(G_i) \le \mathrm{OPT}_{\mathrm{edge}}(G).
\]
Combining the two inequalities proves the claim.
\end{proof}

\section{Polynomial-time algorithm}

We now prove the main theorem.

\begin{proof}[Proof of Theorem~\ref{thm:main}]
If $E(G) = \emptyset$, return any embedding of $G$ and the value $k = 0$. Otherwise, compute the connected components of $G$. By Lemma~\ref{lem:components}, it suffices to solve each nontrivial connected component independently and take the maximum of the resulting values.

Thus assume first that $G$ is connected and has at least one edge. Construct the auxiliary graph $H = H(G)$. This construction is linear:
\[
|V(H)| = |V(G)| + 4|E(G)|, \qquad |E(H)| = 6|E(G)|.
\]
Run the polynomial-time minimum-depth embedding algorithm of Theorem~\ref{thm:adp} on $H$. Let $\Delta^\star$ be a minimum-depth embedding of $H$. From $\Delta^\star$, delete every marker triangle together with its attaching bridge, and suppress every subdivision vertex $s_e$. Let $\Gamma^\star$ be the resulting embedding of $G$.

By Lemma~\ref{lem:backward},
\[
\operatorname{eop}(\Gamma^\star) \le \operatorname{depth}(\Delta^\star).
\]
Since $\Delta^\star$ is a minimum-depth embedding of $H$,
\[
\operatorname{depth}(\Delta^\star) = \mathrm{OPT}_{\mathrm{depth}}(H).
\]
By Theorem~\ref{thm:equality},
\[
\mathrm{OPT}_{\mathrm{depth}}(H) = \mathrm{OPT}_{\mathrm{edge}}(G).
\]
Therefore
\[
\operatorname{eop}(\Gamma^\star) \le \mathrm{OPT}_{\mathrm{edge}}(G).
\]
The reverse inequality holds by the definition of $\mathrm{OPT}_{\mathrm{edge}}(G)$. Hence
\[
\operatorname{eop}(\Gamma^\star) = \mathrm{OPT}_{\mathrm{edge}}(G).
\]
So $\Gamma^\star$ is an embedding of $G$ with minimum possible edge-outerplanarity.

For disconnected $G$, apply the same procedure to each nontrivial connected component and place the resulting component embeddings side by side in a common outer face. By Lemma~\ref{lem:components}, this gives an optimal embedding of $G$.

The running time is polynomial because $H$ has size linear in $|V(G)| + |E(G)|$, and the minimum-depth embedding subroutine runs in polynomial time. Using Theorem~\ref{thm:adp} directly gives an $O(|V(H)|^4)$-time subroutine call for each connected component, hence polynomial time overall.
\end{proof}


\begin{thebibliography}{9}

\bibitem{ADP11}
P.~Angelini, G.~Di Battista, and M.~Patrignani.
Finding a minimum-depth embedding of a planar graph in $O(n^4)$ time.
\emph{Algorithmica}, 60(4):890--937, 2011.
\url{https://doi.org/10.1007/s00453-009-9380-6}.

\bibitem{Bentz09}
C.~Bentz.
Disjoint paths in sparse graphs.
\emph{Discrete Applied Mathematics}, 157(17):3558--3568, 2009.

\bibitem{BM90}
D.~Bienstock and C.~L. Monma.
On the complexity of embedding planar graphs to minimize certain distance measures.
\emph{Algorithmica}, 5:93--109, 1990.
\url{https://doi.org/10.1007/BF01840379}.

\end{thebibliography}
\end{document}